\documentclass{article}
\usepackage{latexsym, a4}
\usepackage[utf8]{inputenc}
\usepackage{graphicx, natbib}
\usepackage{amsmath, amsfonts,amssymb, amsthm}
\usepackage{chemarrow}					

\renewcommand{\cite}{\citet}

\theoremstyle{plain}
\newtheorem{theorem}{Theorem}[section]                                          
\newtheorem{proposition}[theorem]{Proposition}                          
\newtheorem{lemma}[theorem]{Lemma}
\newtheorem{corollary}[theorem]{Corollary}

\theoremstyle{definition}

\theoremstyle{remark}
\newtheorem{remark}[theorem]{Remark}

\newcommand\unnumberedfootnote[1]{ %
        \let\temp=\thefootnote %
        \renewcommand{\thefootnote}{}%
        \footnote{#1}%
        \let\thefootnote=\temp%
        \addtocounter{footnote}{-1}}

\makeatletter

\renewcommand{\@fnsymbol}[1]{\ensuremath{%
    \ifcase#1\or 1\or 2\or 3\or
    \mathsection\or \mathparagraph\or \|\or 1\or
    2\or 3 \else\@ctrerr\fi}}
\makeatother

\begin{document}

\title{Stochastic gene expression with delay}

\thispagestyle{empty}

\author{\sc by Martin Jansen\thanks{University Medical Center
    Freiburg, Institute of Clinical Chemistry and Laboratory Medicine,
    Hugstetter Stra\ss e 55, 79106 Freiburg, Germany, email:
    martin.jansen@uniklinik-freiburg.de} and Peter
  Pfaffelhuber\thanks{Corresponding author; University of Freiburg,
    Abteilung f\"ur Mathematische Stochastik, Eckerstr. 1, 79104
    Freiburg, Germany, email: p.p@stochastik.uni-freiburg.de}}

\maketitle

\unnumberedfootnote{AMS subject classification: 92C42, 92C40, 60K35}
\unnumberedfootnote{Keywords: Chemical reaction network, Poisson process}





\begin{abstract}
  The expression of genes usually follows a two-step procedure. First,
  a gene (encoded in the genome) is transcribed resulting in a strand
  of (messenger) RNA. Afterwards, the RNA is translated into protein.
  We extend the classical stochastic jump model by adding delays (with
  arbitrary distributions) to transcription and translation.

  Already in the classical model, production of RNA and protein come
  in bursts by activation and deactivation of the gene, resulting in a
  large variance of the number of RNA and proteins in equilibrium. We
  derive precise formulas for this second-order structure with the
  model including delay in equilibrium.


\end{abstract}

\maketitle

\section{Introduction}
The central dogma of molecular biology is that a gene (encoded within
the genome) is transcribed into (messenger) RNA (also abbreviated
mRNA), which in turn is translated into protein, the whole process
also being called gene expression. Mathematical models for this
process have by now been studied for a long time; see e.g.\
\cite{Rigney:1977:J-Theor-Biol:607033}, \cite{Berg1978},
\cite{McAdams:1997:Proc-Natl-Acad-Sci-U-S-A:9023339},
\cite{Swain:2002:Proc-Natl-Acad-Sci-U-S-A:12237400},
\cite{Paulsson2005},
\cite{Cottrell:2012:Proc-Natl-Acad-Sci-U-S-A:22660929},
\cite{Bokes2012}, \cite{Pendar2013}, \cite{Robert2013}.

Within a single cell, gene expression often comes with stochastic
fluctuations; see e.g.\ \cite{Raser2005, Raj2008, Balazsi2011}. There
are either one or two copies of the genome, and only a few genes code
for the same protein. As reviewed by \cite{jackson2000balance} the
majority of expressed RNA species in mammalian cells have less than
10 copies, though there are also RNA species present at an order of
10000 copies. \cite{guptasarma1995does} observed that for 80\% of
genes in {\it E. Coli} genome the copy number of many proteins is less
than 100.  Hence in many cases there are only a small copy numbers of
RNA and protein molecules, making them a noisy (i.e.\ stochastic)
quantity. While this stochasticity has been assumed to be detrimental
to the cellular function, it can also help a cell to adapt to
fluctuating environments, or help to explain genetically homogeneous
but phenotypically heterogeneous cellular populations
\citep{Kaern:2005:Nat-Rev-Genet:15883588}.

In order to consider stochasticity in gene expression,
\cite{Swain:2002:Proc-Natl-Acad-Sci-U-S-A:12237400} distinguish
between \emph{intrinsic} and \emph{extrinsic} noise. The latter
accounts for changing environments of the cell, while the former
accounts for the stochastic process of transcription and
translation. Let us look at the possible sources of intrinsic noise in
more detail; see e.g.\ \cite{Zhu:2007:J-Theor-Biol:17350653},
\cite{Roussel:2006:Bull-Math-Biol:16967259}.
\\
(i) Various mechanisms for gene expression require random events to
occur. In order to understand this let us have a closer look at the
mechanisms of gene expression. Transcription starts when RNA
polymerase (which are enzymes helping in the synthetisis of RNA) binds
to the promoter region of the gene, forming an \emph{elongation
  complex}. This elongation complex is then ready to start walking
along the DNA, reading off DNA and making RNA. Before the transcript
is released, a \emph{ribosome binding site} (which is needed for
translation) is being produced on the transcript. Then follows
translation which starts when a free ribosome binds to the ribosome
binding site of the transcript and again is a complex process
involving many chemical reactions, which lead to fluctuations.
\\
(ii) Another source of noise comes from turning genes on and off. This
means that transcription factors can bind to promoter regions of the
gene and only bound (or unbound) promoters can initiate
transcription. This process has been found to be the most important
source of randomness for gene expression (see e.g.\ \citealp{
  Swain:2002:Proc-Natl-Acad-Sci-U-S-A:12237400,
  kaern2005stochasticity, Zhu:2008:FEBS-Lett:18656472, raj2008nature,
  IyerBiswas:2009:Phys-Rev-E-Stat-Nonlin-Soft-Matter-Phys:19391975}). The
effect of this activation and inactivation of genes is a burst-like
behavior of protein production, already apparent in
\cite{McAdams:1997:Proc-Natl-Acad-Sci-U-S-A:9023339}. When considering
the amount of RNA within the cell during the production of a specific
protein, it is hence not surprising that production of RNA comes in
bursts, which are related to times when the gene is turned on. This
burst-like behavior is inherited to protein formation, which also
comes in bursts during translation.

The classical model of stochasticity in gene expression uses
exponential waiting times between transcription and translation
events, and once produced, RNA and protein molecules are immediately
available to the system. The latter contradicts several biological
facts, valid in prokaryotes as well as in eukaryotes, e.g.: Production
of RNA consists of many enzymatic reactions
\citep{Roussel:2006:Bull-Math-Biol:16967259}.  In the translation
process another set of reactions unbinds RNA from the ribosome.  For
eukaryotes, post-transcriptional modification of RNA and the transport
of RNA out of the nucleus to the ribosomes, as well as folding of
proteins, requires time. Taking such issues into account, it makes
sense to model a (random) time delay before an RNA or protein molecule
can be used by the system. In our paper, we are studying the effect of
(random) delay on the noise in gene expression.
In real-life applications, models for such gene expression delays have
been considered e.g.\ by \cite{Lewis:2003:Curr-Biol:12932323},
\cite{Monk:2003:Curr-Biol:12932324},
\cite{Barrio:2006:PLoS-Comput-Biol:16965175},
\cite{Bratsun:2005:Proc-Natl-Acad-Sci-U-S-A:16199522}.

While our modeling approach only takes a single gene/RNA/protein
triple into account, the field of systems biology aims at unraveling
interactions between genes in so-called pathways. It seems clear that
randomness as well as delays can accumulate in such networks of
interacting genes and proteins. As a simple example, the transcription
factor regulating the expression of gene $A$ is coded by a gene $B$
which in turn may be regulated by gene $A$ (or by itself), which can
lead to a bi-modal distribution of the number of proteins encoded by
gene $A$ or $B$; see e.g.\
\cite{Kaern:2005:Nat-Rev-Genet:15883588}. Although such feedback
systems are highly interesting, we are not touching on this level of
complexity.

Today, delays in biochemical reaction networks also serve as a tool
for model reduction. \cite{Barrio:2013:J-Chem-Phys:23514472} and
\cite{Leier2014} argue that lumping together certain reactions
effectively leads to a delay for other reactions. At least for first
order reactions, they compute the resulting delay times which serve
for a precise model reduction.

~

Simulation of chemical systems, or \emph{in silico modeling}, today
paves the way to understanding complex cellular processes. While the
Gillespie algorithm is a classical approach for stochastic simulations
\citep{Gillespie1977} -- see also the review \cite{Gillespie2013} --
chemical delay models have as well been algorithmically studied.
Various explicit simulation schemes for delay models -- in particular
in the field of stochasticity in gene expression -- have been given;
see \cite{Bratsun:2005:Proc-Natl-Acad-Sci-U-S-A:16199522},
\cite{Roussel:2006:Phys-Biol:17200603},
\cite{Barrio:2006:PLoS-Comput-Biol:16965175},
\cite{Cai:2007:J-Chem-Phys:17411109}, \cite{Tian2007},
\cite{Anderson:2007:J-Chem-Phys:18067349},
\cite{Ribeiro:2010:Math-Biosci:19883665},
\cite{Tian:2013:PLoS-One:23349679}, \cite{zavala2014delays}.

The goal of this paper is to give a quantitative evaluation of delay
in the standard model of stochastic gene expression. We do this in a
general way in which the delay -- both for transcription and
translation -- can have an arbitrary distribution. Although we give a
full description of the stochastic processes of the total number of
RNA and protein molecules, our quantitative results are restricted
since we only address the calculation of the first two moments
(expectation, variance and autocovariances) of the number of RNA and
protein molecules.

~

Outline: In Section~\ref{S:model}, we introduce our delay model using
a classical approach of stochastic time-change equations as well as a
description of the system in equilibrium. Then, we present our main
results in Section~\ref{S:results}. Basically, Theorems~\ref{T1}
and~\ref{T2} give the second-order structure of the number of RNA and
protein in equilibrium under the delay model, respectively. In
Section~\ref{S:ex}, we give several examples (uniformly and
exponentially distributed delay, and delay with small variance). We
end our paper with a discussion and connections to previous work in
Section~\ref{S:discuss}.

\section{The model}
\label{S:model}
In order to be able to model gene expression in a sophisticated way,
we now give our delay model.  Using the terminology from
\cite{Roussel1996}, we may write
\begin{equation}
  \label{eq:model}
  \begin{aligned}
    \operatorname{inactive\; gene} &
    \autorightleftharpoons{$\lambda_1^+$}{$\lambda_1^-$}
    \operatorname{active\; gene}\\
    \operatorname{active\; gene}(t) & \autorightarrow{$\lambda_2$}{}
    \operatorname{active\; gene}(t) + \operatorname{RNA\; transcript}(t+G_2)\\
    \operatorname{RNA\; transcript}(t) &
    \autorightarrow{$\lambda_3$}{}
    \operatorname{RNA\; transcript}(t) + \operatorname{protein}(t+G_3)\\
    \operatorname{RNA\; transcript} & \autorightarrow{$1/\tau_2$}{} \emptyset \\
    \operatorname{protein} & \autorightarrow{$1/\tau_3$}{} \emptyset.
  \end{aligned}
\end{equation}
Essentially, \eqref{eq:model} is an extension of the well-studied
model of gene expression, as e.g.\ given in \cite{Paulsson2005}. Gene
expression of $n_{\text{max}}$ similar genes is studied. Each gene is
activated and deactivated at rates $\lambda_1^+$ and $\lambda_1^-$,
respectively. (Additionally, we will set $\tau_1 =
1/(\lambda_1^++\lambda_1^-)$). Every active gene creates the RNA
transcript at rate $\lambda_2$, which is degraded at rate
$1/\tau_2$. However, a RNA molecule is available for the system
(i.e.\ can be translated) only some random delay time $G_2$ after its
creation, where $G_2$ is an independent random variable with
distribution $\mu$. Then, each RNA transcript available for the system
initiates translation of protein at rate $\lambda_3$ which in turn
degrades at rate $1/\tau_3$. Again, it takes a delay of a random time
$G_3$, distributed according to $\nu$ and independent of everything
else, that the protein molecule is available for the system (i.e.\ for
other downstream processes).  

We note that \eqref{eq:model} is a special case of a model studied in
\cite{Zhu:2007:J-Theor-Biol:17350653}. Since they consider the
ribosome binding site as an own chemical species, their model requires
more delay random variables.  Moreover, they distinguish gene
expression in prokaryotes (bacteria) and eukaryotes (higher
organisms), the main difference being that only eukaryotes have a
cellular core. As a consequence, in prokaryotes translation can
already be initiated when transcription is not complete yet. (The
ribosome can bind to the ribosome binding site while the RNA
transcript is still being produced.)  The simplification (6) and (7)
in \cite{Zhu:2007:J-Theor-Biol:17350653} for gene expression in
prokaryotes (both, the time the promoter region of the gene is
occupied and the time the ribosome binds to the ribosome binding site
are negligible), are in line with \eqref{eq:model} for $G_2=0$. In
addition, for the same simplification in eukaryotes (see their
equation (5), where the ribosome binding site is available for binding
to the ribosome only some time after the promoter was released), we
exactly recover \eqref{eq:model} for general $G_2$.

~

The question we ask is about the equilibrium behavior of the number of
available RNA molecules and proteins. We restrict our study to
$n_{\text{max}}=1$, because all genes are independent. We define, with
$t\in\mathbb R$,
\begin{align*}
  \mathcal N_1(t) & = \text{number of active genes (either 0 or 1) at time $t$},\\
  \mathcal N_2(t) & = \text{number of available RNA at time $t$},\\
  \mathcal N_3(t) & = \text{number of available proteins at time $t$}.
\end{align*}
Putting \eqref{eq:model} into well-established time-change equations
\citep{AndersonKurtz2011}, in equilibrium we get
\begin{equation}
  \label{eq:model1}
  \begin{aligned}
    \mathcal N_1(t) & = \mathcal P_1\Big( \int_{-\infty} ^t
    \lambda_1^+ 1_{\mathcal N_1(s)=0} ds\Big) - \mathcal P_2\Big(
    \int_{-\infty} ^t \lambda_1^- 1_{\mathcal N_1(s)=1} ds\Big),\\
    \mathcal N_2(t) & = \mathcal P_3\Big( \int_{-\infty} ^t
    \int_0^\infty \lambda_2 \mathcal N_1(s-\delta) \mu(d\delta)
    ds\Big) - \mathcal P_4\Big( \int_{-\infty} ^t
    \frac{1}{\tau_2} \mathcal N_2(s) ds\Big),\\
    \mathcal N_3(t) & = \mathcal P_5\Big( \int_{-\infty} ^t
    \int_0^\infty \lambda_3 \mathcal N_2(s-\delta') \nu(d\delta')
    ds\Big) - \mathcal P_6\Big( \int_{-\infty} ^t \frac{1}{\tau_3}
    \mathcal N_3(s) ds\Big)
  \end{aligned}
\end{equation}
for independent, unit rate Poisson processes $\mathcal
P_1,...,\mathcal P_6$. (Here, for the measure $\mu$, we use the
standard notation $\mu(d\delta)$ for the mass the measure $\mu$ puts
on the small interval $d\delta$.) Formally, we need to write integrals
$\int_{-T}^tds$ and then obtain the equilibrium by letting
$T\to\infty$. Since we start the process at time $-\infty$, the
initial state is of no relevance due to recurrence of the
process. Also note that \cite{AndersonKurtz2011} indicate that such
delay time-change equations do have a unique solution which is shown
by the same jump by jump argument as for models without delay.
Another description gives the distribution of RNA and protein in
equilibrium. If the gene is active, only some random delay time
$G_2\geq 0$ with distribution $\mu$ later, the RNA is
available. Therefore,
\begin{equation}
  \label{eq:1}
  \begin{aligned}
    \mathcal{N}_2(t) & \sim
    Poi\left\{\lambda_2\int_{0}^{\infty}e^{-\frac{r}{\tau_2}}\int_{0}^{\infty}
      \mathcal{N}_1(t-r-\delta)\mu(d\delta) dr\right\} \\ & =
    Poi\left\{\lambda_2 \tau_2 \int_{0}^{\infty} \mathcal{N}_1(t-s)
      \Big(\exp\Big(\frac{1}{\tau_2}\Big)\ast \mu\Big)(ds)\right\},
  \end{aligned}
\end{equation}
where $\exp(1/\tau)$ is the exponential distribution with expectation
$\tau$ and $\ast$ denotes the convolution of measures.  Indeed, every
RNA available at time $t$ was produced at some time $t-r$, which only
works if at time $t-r-G_2$ (where $G_2\sim\mu$), the gene was
active. In addition, the RNA must not be degraded during time $r$,
which happens with probability $e^{-r/\tau_2}$. Using the same kind of
arguments, we set for another delay $G_3\geq 0$ with distribution
$\nu$
\begin{equation}
  \label{eq:2}
  \begin{aligned}
    \mathcal{N}_3(t)\sim
    &Poi\left\{\lambda_3\int_{0}^{\infty}e^{-\frac{r}{\tau_3}}
      \int_{0}^{\infty}\mathcal{N}_2(t-r-\delta)\nu(d\delta)
      dr\right\} \\ & = Poi\left\{\lambda_3 \tau_3 \int_{0}^{\infty}
      \mathcal{N}_2(t-s) \Big(\exp\Big(\frac{1}{\tau_3}\Big)\ast
      \nu\Big)(ds)\right\}.
  \end{aligned}
\end{equation}
Frequently (see e.g.\ \citealp{Paulsson2005,Bokes2012},) stochasticity
in gene expression is studied through the Master equation. We stress
that the usual Master equation is unsuitable to be used for an
arbitrary delay of RNA and protein production. The reason is simply
that by a non-exponentially distributed delay, the process $(\mathcal
N_1, \mathcal N_2, \mathcal N_3)$ is not a Markov process; however,
compare with~\cite{Tian2007} where an extension of the Master equation
for deterministic delay models is given.  Note that our Poisson point
process approach is similar in spirit to \cite{Robert2013}, who model
an arbitrary (non-exponential) life-time distribution of RNA and
protein using point processes but without delay.

We remark that our modeling is unrealistic at least in one respect:
considering two RNA molecules, created at times $s$ and $s'$ with
$s<s'$, both will be available for translation by times $t = s + G_2$
and $t' = s' + G_2'$. However, since the delay for both RNA molecules
is given through independent $G_2$ and $G_2'$, both having
distribution $\mu$, it might be that $t>t'$.

From~\eqref{eq:model1}, \eqref{eq:1} and \eqref{eq:2}, we already
obtain the equilibrium expectations and second moments in
Proposition~\ref{P1}, Theorem~\ref{T1} and~\ref{T2}. We will see that
expectations are independent of the delay distributions, $\mu$ and
$\nu$. This situation will change for the second-order structure.

\section{Results}
\label{S:results}
\noindent
We are now ready to formulate and prove our results on the number of
active genes (Proposition~\ref{P1}), the amount of RNA molecules
(Theorem~\ref{T1}) and the number of available protein
(Theorem~\ref{T2}). We start with first and second moments of the
active genes.

\begin{proposition}[First and second order structure of $\mathcal
  N_1$]
  The \label{P1} expectation, variance and covariance of $\mathcal N_1$
  in equilibrium are given by
  \begin{align}\label{eq:4}
    \mathbf E[\mathcal N_1(t)] & = \lambda_1^+
    \tau_1, \\
    \label{eq:5}\mathbf{Cov}[\mathcal N_1(0),\mathcal N_1(t)] & =
    e^{-\frac{|t|}{\tau_1}} \lambda_1^+ \lambda_1^-
    \tau_1^2,\\
    \label{eq:6}\mathbf{Var}[\mathcal N_1(t)] & =
    \lambda_1^+ \lambda_1^- \tau_1^2.
  \end{align}
\end{proposition}

\begin{proof}
  Since activation and de-actication of the gene is independent of
  downstream processes, the assertion follows by a simple calculation
  using Poisson processes. We omit the details.
\end{proof}

\noindent
Next, we derive our results for the amount of RNA.

\begin{theorem}[Expectation, variance and covariance of RNA in
  equilibrium]
  The \label{T1} expectation, variance and covariance of $\mathcal
  N_2$ in equilibrium are given by
  \begin{align}
    \label{eq:T10}
    \mathbf E[\mathcal N_2(t)]  & = \lambda_1^+ \lambda_2 \tau_1 \tau_2,\\
    \mathbf{Cov}[ \mathcal N_2(0), \mathcal N_2(t)] & =
    \label{eq:T11} 
    \lambda_1^+ \lambda_2\tau_1\tau_2 e^{-\frac{|t|}{\tau_2}} \\ &
    \notag + \lambda_1^+ \lambda_1^- \lambda_2^2 \tau_1^2
    \tau_2^2 \frac{\tau_1}{\tau_1+\tau_2} \frac{\tau_1 \mathbf
      E\Big[e^{-\frac{|G_2'-G_2-|t||}{\tau_1}}\Big] - \tau_2 \mathbf E
      \Big[e^{-\frac{|G_2'-G_2-|t||}{\tau_2}}\Big]}{\tau_1-\tau_2},\\
    \label{eq:T12} \mathbf{Var}[\mathcal N_2(t)] & = \lambda_1^+
    \lambda_2\tau_1\tau_2 \\ & \qquad + \lambda_1^+ \lambda_1^-
    \lambda_2^2 \tau_1^2 \tau_2^2 \frac{\tau_1}{\tau_1+\tau_2} \cdot
    \frac{\tau_1 \mathbf E\Big[e^{-\frac{|G_2'-G_2|}{\tau_1}}\Big] -
      \tau_2 \mathbf E \notag
      \Big[e^{-\frac{|G_2'-G_2|}{\tau_2}}\Big]}{\tau_1-\tau_2},
  \end{align}
  where $G_2, G'_2$ are independent and $G_2, G'_2\sim\mu$.
\end{theorem}


\begin{remark}[Convexity and deterministic delay]
  We note that the last terms in \eqref{eq:T11} and \eqref{eq:T12}
  stay bounded, even for $\tau_1\to\tau_2$. Moreover, for $X\geq 0$
  the map $t\mapsto \mathbf E[e^{-Xt}]$ is convex. Hence, assuming
  $\tau_1>\tau_2$ without loss of generality, we obtain, using the
  convex combination $1 = \frac{\tau_1}{\tau_1-\tau_2} -
  \frac{\tau_2}{\tau_1-\tau_2}$ (and noting that
  $\frac{\tau_1}{\tau_1-\tau_2}>1$) that for a random variable $S$
  \begin{equation}
    \label{eq:convex}
    \begin{aligned}
      f_S(\tau_1, \tau_2) & := \frac{\tau_1 \mathbf
        E\Big[e^{-\frac{S}{\tau_1}}\Big] - \tau_2 \mathbf E
        \Big[e^{-\frac{S}{\tau_2}}\Big]}{\tau_1-\tau_2} \\
      & \qquad \leq \mathbf E\Big[ \exp\Big(-\Big( \frac{1}{\tau_1}
      \frac{\tau_1}{\tau_1-\tau_2}-\frac{1}{\tau_2}
      \frac{\tau_2}{\tau_1-\tau_2}\Big)S\Big)\Big] = 1
    \end{aligned}
  \end{equation}
  with equality if and only if $\mathbf V[S] = 0$. In particular we
  see that
  \begin{align*}
    \mathbf{Var}[\mathcal N_2(t)] \leq \lambda_1^+
    \lambda_2\tau_1\tau_2 + \lambda_1^+ \lambda_1^- \lambda_2^2
    \tau_1^2 \tau_2^2 \frac{\tau_1}{\tau_1+\tau_2}    
  \end{align*}
  with equality if and only if $\mu$ is a delta-measure.  In
  particular, we see that stochastic delay leads to a decrease in the
  variance for $\mathcal N_2$ in equilibrium. For a deterministic
  delay, we obtain
  \begin{align*}
    \frac{\mathbf{Var}[\mathcal N_2(t)]}{\mathbf E[\mathcal N_2(t)]^2}
    & = \frac{1}{\mathbf E[\mathcal N_2(t)]} +
    \frac{\lambda_1^+\lambda_1^-\lambda_2^2\tau_1^3\tau_2^2}{(\tau_1+\tau_2)\lambda_2^2
      (\lambda_1^+)^2\tau_1^2\tau_2^2} = \frac{1}{\mathbf E[N_2]} +
    \frac{\lambda_1^-}{\lambda_1^+}\frac{\tau_1}{\tau_1 + \tau_2},
  \end{align*}
  which equals the numerical value in the absence of delay, $\mu =
  \delta_0$; see equation~(5) in~\cite{Paulsson2005}.
\end{remark}

\begin{proof}[Proof of Theorem~\ref{T1}]
  Using~\eqref{eq:1} and Proposition~\ref{P1}, we write
  \begin{align*}
    \mathbf E[\mathcal N_2(t)] & = \mathbf E\left[\lambda_2 \tau_2
      \int_{0}^{\infty} \mathcal{N}_1(t-s) \Big(
      \exp\Big(\frac{1}{\tau_2}\Big)\ast \mu\Big)(ds)\right] \\ & =
    \lambda_2 \tau_2 \int_0^\infty \mathbf E[\mathcal{N}_1(t-s)]\Big(
    \exp\Big(\frac{1}{\tau_2}\Big)\ast \mu\Big)(ds) \\ & = \lambda_2
    \tau_2\lambda_1^+\tau_1.
  \end{align*}
  For~\eqref{eq:T11} and~\eqref{eq:T12}, it clearly suffices to
  prove~\eqref{eq:T11}, since the variance formula just requires to
  take $t=0$.  From our model, we can write in the case $t\geq 0$
  (compare with the explanation given below~\eqref{eq:1})
  \begin{align*}
    \mathcal N_2(0) & = X+Y,\\
    \mathcal N_2(t) & = X+Z,
  \end{align*}
  where $X, Y$ and $Z$ are conditionally independent given $\mathcal
  N_1$ such that
  \begin{align*}
    X & \sim
    Poi\left\{\lambda_2\int_{0}^{\infty}e^{-\frac{r+t}{\tau_2}}\int_{0}^{\infty}
      \mathcal{N}_1(-r-\delta)\mu(d\delta) dr\right\} \\ & =
    Poi\left\{\lambda_2 \tau_2 e^{-\frac{t}{\tau_2}} \int_{0}^{\infty}
      \mathcal{N}_1(-s)\Big( \exp\Big(\frac{1}{\tau_2}\Big)\ast
      \mu\Big)(ds) \right\} \intertext{is the number of RNAs available
      by time $0$ which will not be degraded by time $t$, } Y & \sim
    Poi\left\{\lambda_2\int_{0}^{\infty}e^{-\frac{r}{\tau_2}}(1-e^{-\frac{t}{\tau_2}})
      \int_{0}^{\infty} \mathcal{N}_1(-r-\delta)\mu(d\delta)
      dr\right\} \\ & = Poi\left\{\lambda_2 \tau_2
      \big(1-e^{-\frac{t}{\tau_2}}\big) \int_{0}^{\infty}
      \mathcal{N}_1(-s)\Big( \exp\Big(\frac{1}{\tau_2}\Big)\ast
      \mu\Big)(ds) \right\} \intertext{is the number of RNAs available
      by time $0$ which will be degraded by time $t$, } Z & \sim
    Poi\left\{\lambda_2\int_{0}^{t}e^{-\frac{r}{\tau_2}}\int_{0}^{\infty}
      \mathcal{N}_1(t-r-\delta)\mu(d\delta) dr\right\} \\ & =
    Poi\left\{\lambda_2\int_{0}^{\infty }e^{-\frac{r}{\tau_2}}1_{r\leq
        t}\int_{0}^{\infty} \mathcal{N}_1(t-r-\delta)\mu(d\delta)
      dr\right\} \intertext{is the number of RNAs available only
      after time $0$ and present by time $t$.}
  \end{align*}
  Hence, we get
  \begin{align*}
    \mathbf{Cov}[ & \mathcal N_2(0), \mathcal N_2(t)] =
    \mathbf{Cov}[X+Y, X+Z] \\ & = \mathbf E\big[
    \mathbf{Var}[X|\mathcal N_1]\big] + \mathbf{Var}\big[ \mathbf
    E[X|\mathcal N_1]\big] + \mathbf{Cov}\big[\mathbf E[X|\mathcal
    N_1],\mathbf E[Z|\mathcal N_1]\big] \\ & \qquad \qquad \qquad
    \qquad + \mathbf{Cov}\big[\mathbf E[X|\mathcal N_1],\mathbf
    E[Y|\mathcal N_1]\big] + \mathbf{Cov}\big[\mathbf E[Y|\mathcal
    N_1],\mathbf E[Z|\mathcal N_1]\big] \\ & = \mathbf E\Big[\lambda_2
    \tau_2 e^{-\frac{t}{\tau_2}} \int_{0}^{\infty}
    \mathcal{N}_1(-s)\Big( \exp\Big(\frac{1}{\tau_2}\Big)\ast
    \mu\Big)(ds) \Big] \\ & \qquad + \mathbf{Var}\Big[\lambda_2 \tau_2
    e^{-\frac{t}{\tau_2}} \int_{0}^{\infty} \mathcal{N}_1(-s)\Big(
    \exp\Big(\frac{1}{\tau_2}\Big)\ast \mu\Big)(ds) \Big] \\ & \qquad
    + \mathbf{Cov}\Big[\lambda_2 \tau_2 e^{-\frac{t}{\tau_2}}
    \int_{0}^{\infty} \mathcal{N}_1(-s)\Big(
    \exp\Big(\frac{1}{\tau_2}\Big)\ast \mu\Big)(ds), \\ & \qquad
    \qquad \qquad \qquad \qquad \qquad \qquad \qquad
    \lambda_2\int_{0}^{\infty}e^{-\frac{r}{\tau_2}}1_{r\leq
      t}\int_{0}^{\infty} \mathcal{N}_1(t-r-\delta)\mu(d\delta)
    dr\Big] \\ & \qquad + \mathbf{Cov}\Big[\lambda_2 \tau_2
    e^{-\frac{t}{\tau_2}} \int_{0}^{\infty} \mathcal{N}_1(-s)\Big(
    \exp\Big(\frac{1}{\tau_2}\Big)\ast \mu\Big)(ds),\\ & \qquad \qquad
    \qquad \qquad \qquad \qquad \lambda_2 \tau_2
    \big(1-e^{-\frac{t}{\tau_2}}\big) \int_{0}^{\infty}
    \mathcal{N}_1(-s)\Big( \exp\Big(\frac{1}{\tau_2}\Big)\ast
    \mu\Big)(ds)\Big] \\ & \qquad + \mathbf{Cov}\Big[\lambda_2 \tau_2
    \big(1-e^{-\frac{t}{\tau_2}}\big) \int_{0}^{\infty}
    \mathcal{N}_1(-s)\Big( \exp\Big(\frac{1}{\tau_2}\Big)\ast
    \mu\Big)(ds), \\ & \qquad \qquad \qquad \qquad \qquad \qquad
    \qquad \lambda_2\int_{0}^\infty e^{-\frac{r}{\tau_2}}1_{r\leq
      t}\int_{0}^{\infty} \mathcal{N}_1(t-r-\delta)\mu(d\delta)
    dr\Big]\\ & = : A_1 + A_2 + A_3 + A_4 + A_5.
  \end{align*}
  We treat the terms separately and write, using Proposition~\ref{P1},
  for independent, $\exp(1/\tau_2)$-distributed random variables $T_2,
  T_2'$, {\allowdisplaybreaks
    \begin{align*}
      A_1 & =
      \lambda_1^+  \lambda_2\tau_1\tau_2 e^{-\frac{t}{\tau_2}},\\
      A_2 + A_4 & = \lambda_2^2 \tau_2^2\big(e^{-\frac{2t}{\tau_2}} +
      e^{-\frac{t}{\tau_2}}\big(1-e^{-\frac{t}{\tau_2}}\big)
      \int_0^\infty\int_0^\infty \mathbf{Cov}[\mathcal N_1(-s),
      \mathcal N_1(-r)] \\ & \qquad \qquad \qquad \qquad \qquad \qquad
      \qquad \Big( \exp\Big(\frac{1}{\tau_2}\Big)\ast
      \mu\Big)(ds)\Big( \exp\Big(\frac{1}{\tau_2}\Big)\ast
      \mu\Big)(dr) \\ & = \lambda_1^+ \lambda_1^- \lambda_2^2 \tau_1^2
      \tau_2^2 e^{-\frac{t}{\tau_2}} \iint e^{-\frac{|s-r|}{\tau_1}}
      \Big( \exp\Big(\frac{1}{\tau_2}\Big)\ast \mu\Big)(ds)\Big(
      \exp\Big(\frac{1}{\tau_2}\Big)\ast \mu\Big)(dr) \\ & =
      \lambda_1^+ \lambda_1^- \lambda_2^2 \tau_1^2 \tau_2^2
      e^{-\frac{t}{\tau_2}}
      \mathbf E\Big[e^{-\frac{|T_2 + G - T_2'-G'|}{\tau_1}}\Big],\\
      A_3 + A_5 & = \lambda_2^2 \tau_2^2 \iiint 1_{r\leq t}
      \mathbf{Cov}[\mathcal N_1(-s), \mathcal N_1(t-r-\delta)]
      \mu(d\delta)\exp\Big(\frac{1}{\tau_2}\Big)(dr)\\ & \qquad \qquad
      \qquad \qquad \qquad \qquad \qquad \qquad \qquad \qquad \qquad
      \Big( \exp\Big(\frac{1}{\tau_2}\Big)\ast \mu\Big)(ds) \\ & =
      \lambda_1^+ \lambda_1^-\lambda_2^2 \tau_1^2 \tau_2^2 \iiint
      1_{r\leq t} e^{-\frac{|t+s-r-\delta|}{\tau_1}} \mu (d\delta)
      \exp\Big( \frac{1}{\tau_2}\Big)(dr) \\ & \qquad \qquad \qquad
      \qquad \qquad \qquad \qquad \qquad \qquad \qquad \qquad\Big(
      \exp\Big(\frac{1}{\tau_2}\Big)\ast \mu\Big)(ds) \\ & =
      \lambda_1^+ \lambda_1^-\lambda_2^2 \tau_1^2 \tau_2^2 \mathbf
      E\Big[e^{-\frac{|t+T_2+G-T_2'-G'|}{\tau_1}}, T_2'\leq t\Big] \\
      & = \lambda_1^+ \lambda_1^-\lambda_2^2 \tau_1^2 \tau_2^2
      \Big(\mathbf E\Big[e^{-\frac{|t+T_2+G-T_2'-G'|}{\tau_1}}\Big] -
      \mathbf E\Big[e^{-\frac{|t+T_2+G-T_2'-G'|}{\tau_1}}\Big| T_2'>
      t\Big] e^{-\frac{t}{\tau_2}}\Big)\\ & = \lambda_1^+
      \lambda_1^-\lambda_2^2 \tau_1^2 \tau_2^2 \Big(\mathbf
      E\Big[e^{-\frac{|t+T_2+G-T_2'-G'|}{\tau_1}}\Big] - \mathbf
      E\Big[e^{-\frac{|T_2+G-T_2'-G'|}{\tau_1}}\Big]
      e^{-\frac{t}{\tau_2}}\Big)
    \end{align*}
  }because, given $T_2'>t$, the random variable $T_2'-t$ is
  exp$\Big(\frac{1}{\tau_2}\Big)$-distributed. Altogether,
  \begin{align}
    \label{eq:T22b} &\mathbf{Cov}[ \mathcal N_2(0), \mathcal N_2(t)] =
    \lambda_1^+ \lambda_2\tau_1\tau_2 e^{-\frac{|t|}{\tau_2}} +
    \lambda_1^+ \lambda_1^- \lambda_2^2 \tau_1^2 \tau_2^2 \mathbf
    E\Big[ e^{-\frac{||t|+T_2+G-T_2'-G'|}{\tau_1}}\Big].
  \end{align}
  Now, \eqref{eq:T11} follows since by Lemma~\ref{l:key},
  \begin{align*}
    \mathbf E\Big[ & e^{-\frac{||t|+T_2+G-T_2'-G'|}{\tau_1}}\Big] \\ &
    = \frac{\tau_1}{(\tau_1+\tau_2)(\tau_1-\tau_2)} \Big(\tau_1
    \mathbf E\Big[e^{-\frac{|G_2'-G_2-|t||}{\tau_1}}\Big] - \tau_2
    \mathbf E\Big[e^{-\frac{|G_2'-G_2-|t||}{\tau_2}}\Big]\Big).
  \end{align*}
\end{proof}

\bigskip

\begin{theorem}[Expectation and variance of protein in equilibrium]
  The expectation and \label{T2} variance of $\mathcal N_3$ in
  equilibrium are given by
  \begin{align}
    \label{eq:T20}
    & \mathbf E[\mathcal N_3(t)] = \lambda_1^+\lambda_2\lambda_3
    \tau_1 \tau_2
    \tau_3,\\
    \label{eq:T21}
    & \mathbf{Var}[\mathcal N_3(t)] = \lambda_1^+ \lambda_2 \lambda_3
    \tau_1\tau_2 \tau_3 + \lambda_1^+ \lambda_2 \lambda_3^2 \tau_1
    \tau_2 \tau_3^2 \cdot A + \lambda_1^+ \lambda_1^- \lambda_2^2
    \lambda_3^2 \tau_1^2
    \tau_2^2 \tau_3^2 \cdot B,\\
    \label{eq:T22}
    & A = \frac{\tau_2}{\tau_2 + \tau_3} \frac{\tau_2 \mathbf
      E\Big[e^{-\frac{|G_3 - G_3'|}{\tau_2}}\Big]
      - \tau_3 \mathbf E\Big[e^{-\frac{|G_3 - G_3'|}{\tau_3}}\Big]}{\tau_2 - \tau_3}\\
    \label{eq:T23}
    & B = 
    \frac{\tau_1}{\tau_1 + \tau_3} \Big( \frac{\tau_1^2}{\tau_1 +
      \tau_2} \frac{\tau_1 \mathbf E\Big[e^{-\frac{|G_3-G_3' + G_2 -
          G_2'|}{\tau_1}}\Big] - \tau_2\mathbf
      E\Big[e^{-\frac{|G_3-G_3' + G_2 -
          G_2'|}{\tau_2}}\Big]}{(\tau_1 - \tau_3)(\tau_1 - \tau_2)} \\
    \notag & \qquad \qquad \qquad \qquad - \frac{\tau_3^2}{\tau_2 +
      \tau_3} \frac{\tau_2 \mathbf E\Big[e^{-\frac{|G_3-G_3' + G_2 -
          G_2'|}{\tau_2}}\Big] - \tau_3\mathbf
      E\Big[e^{-\frac{|G_3-G_3' + G_2 - G_2'|}{\tau_3}}\Big]}{(\tau_1
      - \tau_3)(\tau_2 - \tau_3)} \Big),
  \end{align}
  where $G_2, G_2', G_3, G_3'$ and independent and $G_2, G_2' \sim
  \mu$ and $G_3, G_3'\sim\nu$.
\end{theorem}

\begin{remark}[Convexity and deterministic delay]
  We \label{rem:convex2} note that the last terms in \eqref{eq:T22}
  and \eqref{eq:T23} stay bounded, even for $\tau_1\to\tau_2$,
  $\tau_2\to\tau_3$ and $\tau_1 \to \tau_3$. In addition, $A, B\geq 0$
  can easily be shown, which means that $\mathbf{Var}[\mathcal
  N_3(t)]\geq \mathbf E[\mathcal N_3(t)]$ in all cases. Moreover, we
  will -- as for the second moments of the amount of RNA -- argue that
  $\mathbf{Var}[\mathcal N_3(t)]$ is largest in the absence of
  delay. First, recall from~\eqref{eq:convex} that $A\leq 1$ with
  equality if and only if $\mu$ is a delta-measure (i.e.\
  $\mathbf{Var}[G_2]=0$). Moreover, for a similar bound on $B$, assume
  without loss of generality that $\tau_1>\tau_3$ and write, using
  again~\eqref{eq:convex} and $S:=|G_3-G_3'+G_2-G_2'|$
  \begin{align*}
    \frac{\tau_1 + \tau_3}{\tau_1} B & = \frac{\tau_1^2f_S(\tau_1,
      \tau_2)}{(\tau_1 + \tau_2)(\tau_1 - \tau_3)} - \frac{\tau_3^2
      f_S(\tau_2, \tau_3)}{(\tau_2 + \tau_3)(\tau_1 - \tau_3)} \\ & =
    \frac{\tau_1^2 \tau_2f_S(\tau_1, \tau_2) + \tau_1^2
      \tau_3f_S(\tau_1, \tau_2) - \tau_1\tau_3^2f_S(\tau_2, \tau_3) -
      \tau_2 \tau_3^2 f_S(\tau_2,
      \tau_3)}{(\tau_1 + \tau_2)(\tau_2 + \tau_3)(\tau_1 - \tau_3)} \\
    & = \frac{\tau_2(\tau_1 + \tau_3)}{(\tau_1 + \tau_2)(\tau_2 +
      \tau_3)} \frac{\tau_1^2f_S(\tau_1,
      \tau_2) - \tau_3^2 f_S(\tau_3, \tau_2)}{\tau_1^2 - \tau_3^2} \\
    & \qquad \qquad \qquad \qquad + \frac{\tau_1 \tau_3}{(\tau_1 +
      \tau_2)(\tau_2 + \tau_3)}\frac{\tau_1f_S(\tau_1, \tau_2) -
      \tau_3 f_S(\tau_3, \tau_2)}{\tau_1 - \tau_3}.
  \end{align*}
  Next, it is easy to check that the function $t\mapsto f_S(t,
  \tau_2)$ is convex. (For this, it suffices to show that $t\mapsto
  t\mathbf E[e^{-S/t}]$ is convex, which can be shown by computing two
  derivatives.) Therefore, using the convex combinations $1 =
  \frac{\tau_1^2}{\tau_1^2 - \tau_3^2} - \frac{\tau_1^3}{\tau_1^2 -
    \tau_3^2}$ and $1 = \frac{\tau_1}{\tau_1 - \tau_3} -
  \frac{\tau_3}{\tau_1 - \tau_3}$ (recall $\tau_1 > \tau_3$), we see
  that
  \begin{align*}
    & \frac{\tau_1^2f_S(\tau_1, \tau_2) - \tau_3^2 f_S(\tau_3,
      \tau_2)}{\tau_1^2 - \tau_3^2} \\ & \qquad \leq f_S\Big(
    \frac{\tau_3^2}{\tau_1^2 - \tau_3^2} - \frac{\tau_3^2}{\tau_1^2 -
      \tau_3^2}, \tau_2\Big) = f_S\Big(\frac{\tau_1^2 + \tau_1\tau_3 +
      \tau_3^2}{\tau_1 + \tau_3}, \tau_2\Big)\leq 1,\\
    & \frac{\tau_1f_S(\tau_1, \tau_2) - \tau_3 f_S(\tau_3,
      \tau_2)}{\tau_1 - \tau_3} \\ & \qquad \leq f_S\Big(
    \frac{\tau_1^2}{\tau_1 - \tau_3} - \frac{\tau_3^2}{\tau_1 -
      \tau_3}, \tau_2\Big) = f_S(\tau_1 + \tau_3, \tau_2)\leq 1
  \end{align*}
  with equality if and only if $\mathbf{Var}[S]=0$. In total, we see
  that
  \begin{align*}
    \mathbf{Var}[\mathcal N_3(t)] & \leq \lambda_1^+ \lambda_2
    \lambda_3 \tau_1 \tau_2 \tau_3 + \lambda_1^+ \lambda_2 \lambda_3^2
    \tau_1 \tau_2\tau_3^2 \frac{\tau_2}{\tau_2 + \tau_3} \\ & \qquad
    \qquad \qquad \qquad + \lambda_1^+ \lambda_1^-
    \lambda_2^2\lambda_3^2 \tau_1^2 \tau_2^2 \tau_3^2
    \frac{\tau_1}{\tau_1 + \tau_3} \frac{\tau_1 \tau_2 + \tau_1
      \tau_3+ \tau_2 \tau_3 }{(\tau_2 + \tau_3)(\tau_1 + \tau_2)}
  \end{align*}
  with equality if $\mu$ and $\nu$ are delta-measures. In this case,
  \begin{align*}
    \frac{\mathbf{Var}[\mathcal N_3(t)]}{\mathbf E[\mathcal N_3(t)]^2}
    & = \frac{1}{\mathbf E[\mathcal N_3(t)]} + \frac{1}{\lambda_1^+
      \lambda_2 \tau_1 \tau_2} \frac{\tau_2}{\tau_2 + \tau_3} +
    \frac{\lambda_1^-}{\lambda_1^+} \frac{\tau_1}{\tau_1 +
      \tau_2}\frac{\tau_1 \tau_2 + \tau_1 \tau_3
      +\tau_2\tau_3}{(\tau_2 + \tau_3)(\tau_1 + \tau_3)},
  \end{align*}
  which is a well-known result; see equation~(4)
  in~\cite{Paulsson2005}.
\end{remark}


\begin{proof}[Proof of Theorem~\ref{T2}]
  Using~\eqref{eq:2} and Theorem~\ref{T1}
  \begin{align*}
    \mathbf E[\mathcal N_3(t)] & = \mathbf E\left[\lambda_3 \tau_3
      \int_{0}^{\infty}\mathcal{N}_2(t-s)\Big(\exp\Big(\frac{1}{\tau_3}\Big)
      \ast \nu\Big)(ds)\right] \\ &= \lambda_3\tau_3 \lambda_2
    \lambda_1^+\tau_1\tau_2.
  \end{align*}
  For the variance, we again use~\eqref{eq:2} and two independent
  $\exp(1/\tau_3)$-distributed random variables $T_3, T_3'$ for
  \begin{align}
    \notag & \mathbf{Var}[\mathcal N_3(t)] = \mathbf
    E\big[\mathbf{Var}[\mathcal N_3(t)|\mathcal N_1, \mathcal
    N_2]\big] + \mathbf{Var}\big[\mathbf E[\mathcal N_3(t)|\mathcal
    N_1, \mathcal N_2]\big], \intertext{with} \notag & \mathbf
    E\big[\mathbf{Var}[\mathcal N_3(t)|\mathcal N_1, \mathcal
    N_2]\big] = \mathbf{E}\Big[\lambda_3 \tau_3 \int_0^\infty \mathcal
    N_2(t-r) \Big(\exp\Big(-\frac 1{\tau_3}\Big)\ast \nu\Big)(dr)\Big]
    \\ \notag & \qquad \qquad \qquad \qquad \qquad =
    \lambda_1^+\lambda_2 \lambda_3\tau_1\tau_2\tau_3 \intertext{and
      with~\eqref{eq:T22b} in the third equality} \notag &
    \mathbf{Var}\big[\mathbf E[\mathcal N_3(t)|\mathcal N_1, \mathcal
    N_2]\big] = \mathbf{Var}\Big[\lambda_3 \tau_3 \int_0^\infty
    \mathcal N_2(t-r)
    \Big(\exp\Big(\frac 1{\tau_3}\Big)\ast \nu\Big)(dr)\Big] \\
    \notag & = \lambda_3^2 \tau_3^2 \int_0^\infty\int_0^\infty
    \mathbf{Cov}[\mathcal N_2(t-s), \mathcal N_2(t-r)] \\ \notag &
    \qquad \qquad \qquad \qquad \qquad \Big(\exp\Big(\frac
    1{\tau_3}\Big)\ast \nu\Big)(dr)\Big(\exp\Big(\frac
    1{\tau_3}\Big)\ast \nu\Big)(ds) \\ \notag & = \lambda_1^+
    \lambda_2\lambda_3^2 \tau_1\tau_2 \tau_3^2 \int_0^\infty
    \int_0^\infty e^{-\frac{|s-r|}{\tau_2}} \Big(\exp\Big(\frac
    1{\tau_3}\Big)\ast \nu\Big)(dr)\Big(\exp\Big(\frac
    1{\tau_3}\Big)\ast \nu\Big)(ds) \\ \notag & \qquad + \lambda_1^+
    \lambda_1^-\lambda_2^2\lambda_3^2\tau_1^2 \tau_2^2 \tau_3^2 \\
    \notag & \qquad \qquad \Big( \int_0^\infty \int_0^\infty \mathbf
    E\Big[e^{-\frac{||s-r| + T_2+G_2-T_2'-G_2'|}{\tau_1}}\Big] \\
    \notag & \qquad \qquad \qquad \qquad \qquad \qquad
    \Big(\exp\Big(\frac 1{\tau_3}\Big)\ast
    \nu\Big)(dr)\Big(\exp\Big(\frac 1{\tau_3}\Big)\ast
    \nu\Big)(ds)\Big) \\ \label{eq:651} & = \lambda_1^+
    \lambda_2\lambda_3^2 \tau_1\tau_2 \tau_3^2 \mathbf
    E\Big[e^{-\frac{|T_3+G_3 - T_3'-G_3'|}{\tau_2}}\Big] \\ \notag &
    \qquad + \lambda_1^+ \lambda_1^-\lambda_2^2\lambda_3^2\tau_1^2
    \tau_2^2 \tau_3^2 \mathbf E\Big[e^{-\frac{||T_3+G_3-T_3'-G_3'| +
        T_2+G_2-T_2'-G_2'|}{\tau_1}}\Big].
  \end{align}
  Now, by Lemma~\ref{l:key}, \eqref{eq:key1}, 
  \begin{align}
    \label{eq:652}
    \mathbf E\Big[e^{-\frac{|T_3+G_3-T_3'-G_3'|}{\tau_2}}\Big] & =
    \frac{\tau_2}{\tau_2 + \tau_3} \frac{\tau_2\mathbf
      E\Big[e^{-\frac{|G_3-G_3'|}{\tau_2}}\Big] - \tau_3\mathbf
      E\Big[e^{-\frac{|G_3-G_3'|}{\tau_3}}\Big]}{\tau_2 - \tau_3}
  \end{align}
  and by applying \eqref{eq:key2} twice (first for $S=T_2 + G_2' - T_2
  - G_2$ and then for $S=G_2 - G_2' + G_3 - G_3'$
  \begin{align}
    \label{eq:653}
    \mathbf E\Big[ & e^{-\frac{||T_3+G_3-T_3'-G_3'| +
        T_2+G_2-T_2'-G_2'|}{\tau_1}}\Big] \\ \notag & =
    \frac{\tau_1}{\tau_1 + \tau_3} \frac{\tau_1 \mathbf
      E\Big[e^{-\frac{|T_2+G_2-T_2'-G_2' + G_3 - G_3'|}{\tau_1}}\Big]
      - \tau_3 \mathbf E\Big[e^{-\frac{|T_2+G_2-T_2'-G_2' + G_3 -
          G_3'|}{\tau_3}}\Big]}{\tau_1 - \tau_3} \\ \notag & =
    \frac{\tau_1}{\tau_1 + \tau_3} \Big( \frac{\tau_1^2}{\tau_1 +
      \tau_2} \frac{\tau_1 \mathbf E\Big[e^{-\frac{|G_3-G_3' + G_2 -
          G_2'|}{\tau_1}}\Big] - \tau_2\mathbf
      E\Big[e^{-\frac{|G_3-G_3' + G_2 -
          G_2'|}{\tau_2}}\Big]}{(\tau_1 - \tau_3)(\tau_1 - \tau_2)} \\
    \notag & \qquad \qquad \qquad - \frac{\tau_3^2}{\tau_2 + \tau_3}
    \frac{\tau_2 \mathbf E\Big[e^{-\frac{|G_3-G_3' + G_2 -
          G_2'|}{\tau_2}}\Big] - \tau_3\mathbf
      E\Big[e^{-\frac{|G_3-G_3' + G_2 - G_2'|}{\tau_3}}\Big]}{(\tau_1
      - \tau_3)(\tau_2 - \tau_3)} \Big).
  \end{align}
  Hence, plugging \eqref{eq:652} and \eqref{eq:653} in \eqref{eq:651},
  we get the result.
\end{proof}

\section{Examples}
\label{S:ex}
Here, we present some examples for different kinds of delays and their
consequences on the variances of the number of RNA and protein
molecules, respectively. The first two are uniform and exponential
delays, which are also compared in Figure~\ref{fig1}. The main work is
to compute the quantities $A$ and $B$ from Theorem~\ref{T2}. We also
present a result for a delay of small variance in
Subsection~\ref{ss:small}.

\begin{figure}
  \begin{center}
    \includegraphics[width=6cm]{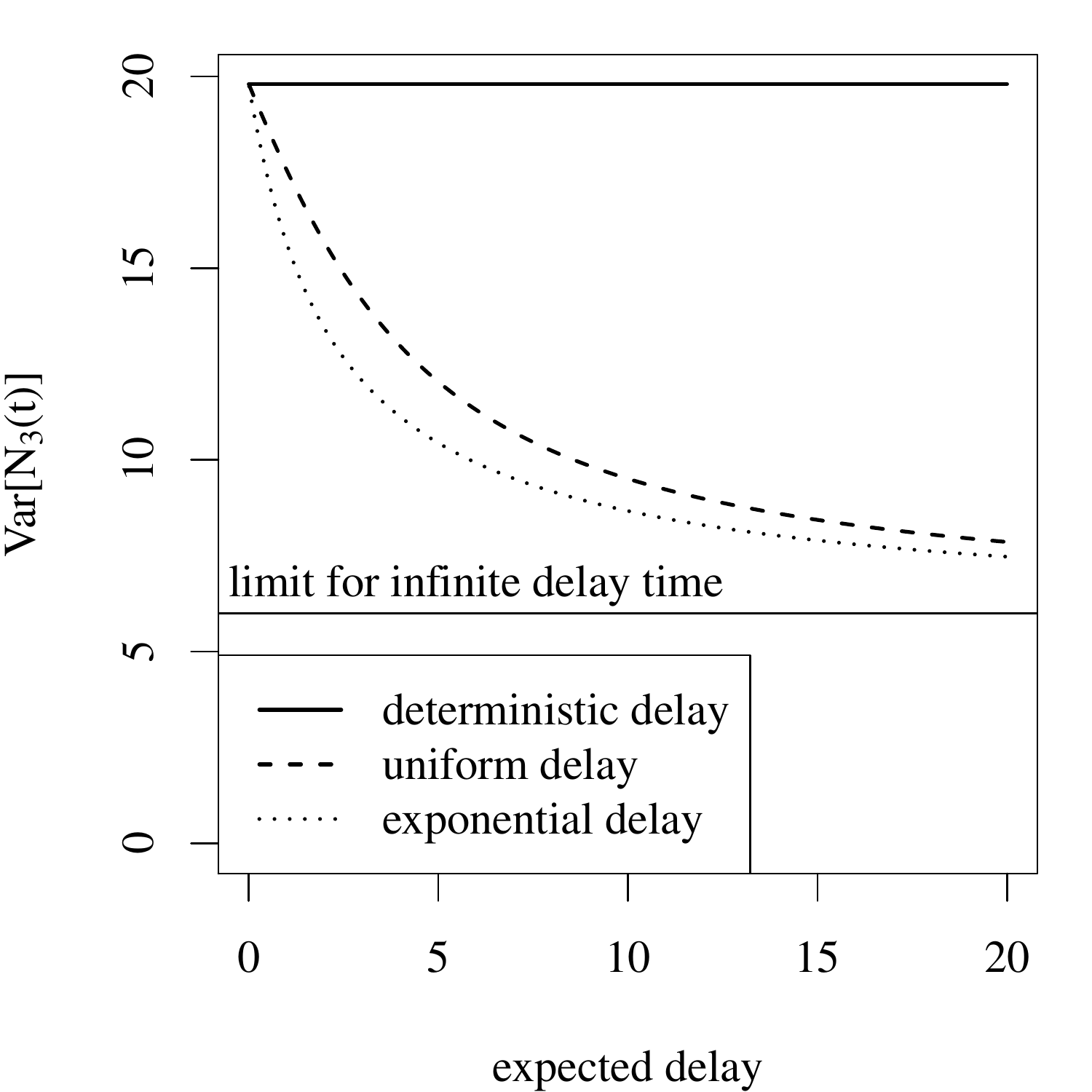}
    \caption{\label{fig1} The results from Theorem~\ref{T2},
      Lemma~\ref{l:uni} and Lemma~\ref{l:exp} are summarized in this
      figure. We fix $\lambda_1^+=\lambda_1^-=\lambda_2 =
      \lambda_3=1$, $\tau_1=1, \tau_2=2, \tau_3=3$ here. We display
      the dependency of the variance of the number of proteins
      $\mathbf {Var}[\mathcal N_3(t)]$, on the expected delay
      time. For deterministic delay, the variance does not change in
      equilibrium, while it decreases for uniform and even more for
      exponentially distributed delay. For simplicity, we choose $b=d$
      and $a=c=0$ for Lemma~\ref{l:uni} and $\sigma_2 = \sigma_3$ in
      Lemma~\ref{l:exp}.}
  \end{center}
\end{figure}

\subsection{Uniform delay}
\begin{lemma}[Uniform delay\label{l:uni}]
  Let $\mu$ be the uniform distribution on $[a,b]$, $\nu$ be the
  uniform distribution on $[c,d]$ and $G_2,G_2'\sim\mu, G_3,
  G_3'\sim\nu$ be independent. Then, for $e_+ := max(d-c, b-a)$ and
  $e_-:=\min(d-c,b-a)$,
  \begin{align*}
    \mathbf E\Big[e^{-\frac{|G_2 - G_2'|}{\tau}}\Big] & =
    \frac{2\tau}{b-a}\Big( 1 -
    \frac{\tau}{b-a}(1-e^{-\frac{b-a}{\tau}})\Big),\\
    \mathbf E\Big[e^{-\frac{|G_3-G_3'+G_2 - G_2'|}{\tau}}\Big] & =
    \frac{2\tau}{3(b-a)^2(d-c)^2}\Big(3 e_-^2 e_+ - e_-^3 - 6 e_-
    \tau^2 \\ & \qquad \qquad \qquad \qquad + 3
    e^{-\frac{e_-+e_+}{\tau}} (e^{\frac{e_-}{\tau}}-1)
    (e^{\frac{e_-}{\tau}} + 2 e^{\frac{e_+}{\tau}}-1) \tau^3\Big).
  \end{align*}
\end{lemma}

\begin{proof}
  Without loss of generality we can assume that $a=c=0$ since
  deterministic delays do not affect our result. Let $G_2,
  G_2'\sim\mu$ and $G_3, G_3'\sim\nu$ be independent. Note that the
  density of $\mu\ast\nu \sim G_2 + G_3$ is given by (recall $s\wedge
  t := \min(s,t)$ and $s^+ := \max(s,0)$)
  \begin{align*}
    f(x) & = \frac{1}{bd} \int_0^x 1_{y\leq b} 1_{x-y \leq d} dy =
    \frac{1}{bd} \int_{(x-d)^+}^{x\wedge b} dy = \frac{(x\wedge b) -
      (x-d)^+ }{bd}
  \end{align*}
  for $0\leq x\leq b+d$. We need to compute (using Mathematica),
  assuming $b\leq d$,
  \begin{align*}
    & \mathbf E\Big[e^{-\frac{|G_2 - G_2'|}{\tau}}\Big] =
    \frac{2}{b^2} \int_0^b \int_0^x e^{\frac{y-x}{\tau}} dy dx =
    \frac{2\tau}{b^2}\int_0^b 1 - e^{-\frac x\tau} dx \\ & \qquad=
    \frac{2\tau}{b}\Big( 1 -
    \frac{\tau}{b}(1-e^{-\frac{b}{\tau}})\Big),\\
    & \mathbf E\Big[e^{-\frac{|G_3-G_3'+G_2 - G_2'|}{\tau}}\Big] =
    \frac{2}{b^2d^2} \int_0^{b+d} \int_0^x e^{\frac{y-x}{\tau}}
    ((x\wedge b) - (x-d)^+)\\ & \qquad\qquad\qquad\qquad\qquad\qquad\qquad
    \cdot ((y\wedge b) - (y-d)^+) dy dx \\ & = \frac{2}{b^2d^2} \Big(
    \int_0^{b} \int_0^x e^{\frac{y-x}{\tau}}xy dy dx +
    \int_b^{b+d}\int_0^b e^{\frac{y-x}{\tau}} by dy \\ & \qquad\qquad +
    \int_b^{b+d}\int_b^x e^{\frac{y-x}{\tau}} b^2 dy dx -
    \int_d^{b+d}\int_d^x e^{\frac{y-x}{\tau}} b(y-d) dy dx \\ & \qquad\qquad
    - \int_d^{b+d}\int_0^b e^{\frac{y-x}{\tau}} (x-d)y dy dx-
    \int_d^{b+d}\int_b^x e^{\frac{y-x}{\tau}} (x-d)b dy dx \\ & \qquad\qquad
    +
    \int_d^{b+d} \int_d^x e^{\frac{y-x}{\tau}}(x-d)(y-d) dy dx\Big) \\\qquad
    & = \frac{2\tau}{3b^2d^2}\Big(3 b^2 d - b^3 + - 6 b \tau^2 + 3
    e^{-\frac{b + d}{\tau}} (e^{\frac{b}{\tau}}-1) (e^{\frac{b}{\tau}}
    + 2
    e^{\frac{d}{\tau}}-1) \tau^3\Big).
    \end{align*} 
\end{proof}

\subsection{Exponential delay}
\begin{lemma}[Exponential delay\label{l:exp}]
  For $\sigma_2, \sigma_3>0$, let $\mu$ be the
  $\exp(1/\sigma_2)$-distribution, $\nu$ be the
  $\exp(1/\sigma_3)$-distribution and $G_2,G_2'\sim\mu, G_3,
  G_3'\sim\nu$ be independent. Then,
  \begin{align*}
    \mathbf E\Big[e^{-\frac{|G_2 - G_2'|}{\tau}}\Big] & =
    \frac{\tau}{\sigma_2 + \tau},\\
    \mathbf E\Big[e^{-\frac{|G_3-G_3'+G_2 - G_2'|}{\tau}}\Big] & =
    \frac{\tau(\sigma_2\sigma_3 + \tau(\sigma_2 +
      \sigma_3))}{(\tau+\sigma_2)(\tau+\sigma_3)(\sigma_2 +
      \sigma_3)}.
  \end{align*}
\end{lemma}

\begin{proof}
  Note that the distribution of $|G_2-G_2'|$ has the density
  $$ f(x) = \frac{2}{\sigma_2^2}\int_x^\infty e^{-\frac{y}{\sigma_2}} e^{-\frac{y-x}{\sigma_2}}dy = \frac{1}{\sigma_2} 
  e^{-\frac{x}{\sigma_2}},$$ i.e.\ is again an
  $\exp(1/\sigma_2)$-distribution. Hence,
  \begin{align*}
    \mathbf E\Big[e^{-\frac{|G_2 - G'_2|}{\tau}}\Big] & =
    \frac{\tau}{\sigma_2 + \tau}.
  \end{align*}
  Moreover, we can use Lemma~\ref{l:key} with $t=0$, $G=G_2, G'=G_2'$,
  $T=G_3, T'=G_3'$, $\sigma = \sigma_3$, in order to see that
  \begin{align*}
    \mathbf E\Big[e^{-\frac{|G_2 - G'_2 + G_3 - G_3'|}{\tau}}\Big] & =
    \frac{\tau}{\tau + \sigma_3}\frac{\tau\mathbf E\Big[e^{-\frac{|G_2
          - G'_2|}{\tau}}\Big] - \sigma_3\mathbf E\Big[e^{-\frac{|G_2
          - G'_2|}{\sigma_3}}\Big]}{\tau - \sigma_3} \\ & =
    \frac{\tau}{\tau+\sigma_3}\frac{\tau^2(\sigma_2 + \sigma_3) -
      \sigma_3^2(\sigma_2 + \tau)}{(\tau-\sigma_3)(\sigma_2 +
      \tau)(\sigma_2 + \sigma_3)} \\ & = \frac{\tau(\sigma_2\sigma_3 +
      \tau(\sigma_2 +
      \sigma_3))}{(\tau+\sigma_2)(\tau+\sigma_3)(\sigma_2 + \sigma_3)}.
  \end{align*}
\end{proof}

\subsection{Small variance}
\label{ss:small}
Delays can be the result of various mechanisms; see
\cite{Barrio:2006:PLoS-Comput-Biol:16965175} and references therein
for a list of possible mechanisms. Hence, by the central limit
theorem, it is reasonable to assume that the delay distribution has a
small variance. For this case, we obtain the following result.

\begin{corollary}
  Let $G_2, G_2'\sim \mu$ and $G_3, G_3'\sim\nu$ are independent, and
  such that $\mathbf{Var}[G_2], \mathbf{Var}[G_3]$ is small. Then, if
  $$\delta_2 := \mathbf E[|G_2 - \mathbf E[G_2]|^3], \qquad \delta_3 := 
  \mathbf E[|G_3 - \mathbf E[G_3]|^3],$$ it holds for $\delta_2,
  \delta_3\to 0$ that
  \begin{align*}
    \mathbf{Var}[\mathcal N_2(t)] & = \lambda_1^+ \lambda_2 \tau_1 \tau_2
    + \lambda_1^+\lambda_1^- \lambda_2^2 \tau_1^2 \tau_2
    \frac{1}{\tau_1 + \tau_2}\Big(1 - \frac{2\mathbf{Var}[G_2]}{\tau_1
      \tau_2}\Big) + \mathcal
    O(\delta_2),\\
    \mathbf{Var}[\mathcal N_3(t)] & = \lambda_1^+ \lambda_2 \lambda_3
    \tau_1 \tau_2 \tau_3 + \lambda_1^+ \lambda_2 \lambda_3^2
    \frac{\tau_1 \tau_2^2 \tau_3^2}{\tau_2 + \tau_3} \Big(1 -
    \frac{2\mathbf{Var}[G_3]}{\tau_2\tau_3}\Big) \\ & \qquad +
    \lambda_1^+ \lambda_1^- \lambda_2^2\lambda_3^2 \tau_1^3 \tau_2^2
    \tau_3^2\frac{\tau_1 \tau_2 + \tau_1 \tau_3 + \tau_2 \tau_3 -
      \mathbf{Var}[G_2 + G_3]}{(\tau_1 + \tau_2)(\tau_1 + \tau_3)( \tau_2
      + \tau_3)}+ \mathcal O(\delta_2, \delta_3).
  \end{align*}
\end{corollary}

\begin{proof}
  We write
  \begin{align*}
    \mathbf E\big[e^{-\frac{|G_2'-G_2|}{\tau}}\big] & = \mathbf
    E\Big[1 - \frac{|G_2'-G_2|}{\tau} + \frac{(G_2'-G_2)^2}{\tau^2} +
    \mathcal O\Big(|G_2'-G_2|^3\Big)\Big] \\ & = \mathbf E\Big[1 -
    \frac{|G_2'-G_2|}{\tau}+ \mathcal O\Big(|G_2'-G_2|^3\Big)\Big] +
    \frac{2\mathbf{Var}[G_2]}{\tau^2}.
  \end{align*}
  Hence, the first result concerning $\mathbf{Var}[\mathcal N_2(t)]$
  follows directly from Theorem~\ref{T1}. For the variance of
  proteins,
  \begin{align*}
    \mathbf E\Big[ & e^{-\frac{|G_3-G_3' + G_2 - G_2'|}{\tau}}\Big] \\
    & = \mathbf E\Big[1 - \frac{|G_3 - G_3' + G_2-G_2'|}{\tau} +
    \frac{(G_3-G_3'+G_2-G_2')^2}{\tau^2} \\ & \qquad \qquad \qquad
    \qquad \qquad \qquad \qquad \qquad + \mathcal O\Big(|G_2-G_2'|^3,
    |G_3-G_3'|^3\Big)\Big] \\ & = \mathbf E\Big[1 - \frac{|G_3 - G_3'
      + G_2-G_2'|}{\tau}\Big] + \frac{2(\mathbf{Var}[G_2] +
      \mathbf{Var}[G_3])}{\tau^2} + \mathcal O(\delta_2, \delta_3).
  \end{align*}
  So, denoting the values of $A$ and $B$ from \eqref{eq:T23} with
  deterministic $G_2$ and $G_3$ by $A_0$ and $B_0$, respectively
  (compare with Remark~\ref{rem:convex2}), \eqref{eq:T23} gives
  \begin{align*}
    A_0 - A & = \frac{2\mathbf{Var}[G_3]}{(\tau_2 + \tau_3)\tau_3},\\
    B_0 - B & = 
    \frac{\tau_1}{(\tau_1 + \tau_3)(\tau_1 - \tau_3)}
    \Big(\frac{2\mathbf{Var}[G_2 + G_3]\tau_1}{(\tau_1 + \tau_2)\tau_2}
    -\frac{2\mathbf{Var}[G_2 + G_3]\tau_3}{(\tau_2 + \tau_3)\tau_2}\Big)
    + \mathcal O(\delta_2, \delta_3)
    \\ & = 
    \frac{2\mathbf{Var}[G_2 + G_3]\tau_1}{(\tau_1 + \tau_2)(\tau_1 +
      \tau_3)(\tau_2 + \tau_3)} + \mathcal O(\delta_2, \delta_3).
  \end{align*}
  Hence, if $\mathbf{Var}$ is the variance of $\mathcal
  N_3(t)$ for deterministic $G_2, G_3$,
  \begin{align*}
    \mathbf{Var}[\mathcal N_3(t)] & = \mathbf{Var} -2 \Big(
    \lambda_1^+\lambda_2 \lambda_3^2 \tau_1 \tau_2 \tau_3
    \frac{\mathbf{Var}[G_3]}{\tau_2 + \tau_3} \\ & \qquad \qquad +
    \lambda_1^+ \lambda_1^- \lambda_2^2 \lambda_3^2\tau_1^3 \tau_2^2
    \tau_3^2\frac{\mathbf{Var}[G_2 + G_3]}{(\tau_1 + \tau_2)(\tau_1 +
      \tau_3)(\tau_2 + \tau_3)}\Big) + \mathcal O(\delta_2, \delta_3)
  \end{align*}
  and the result follows.
\end{proof}

\section{Discussion}
\label{S:discuss}
Our main results, Theorems~\ref{T1} and \ref{T2} give precise formulas
on the first two moments of the number of RNA and protein in
equilibrium for the delay model considered in~\eqref{eq:model}. We
show that the expectation is not influenced by the delay but the
variance tends to be highest without delay. As seen in Theoren
\ref{T1} the variance $\mathbf{Var}[\mathcal N_2(t)]$ can be
decomposed into the sum of
\begin{itemize}
\item $\mathbf{E}[\mathbf{Var}[\mathcal N_2(t)|\mathcal
  N_1]] = \mathbf E[\mathcal N_2(t)] = \lambda_1^+\lambda_2 \tau_1\tau_2$
\item $\mathbf{Var}[\mathbf{E}[\mathcal N_2(t)|\mathcal N_1]] =
  \lambda_1^+\lambda_1^-\lambda_2^2\tau_1^2\tau_2^2\cdot C$
\end{itemize}
with $C$ seen from \eqref{eq:T12}. Here the first term can be
interpreted as individual RNA-part and the second term as noise due to
gene-activation-part. A similar decomposition also holds for the
variance of protein number into (compare with
\citealp{bowsher2012identifying})
\begin{itemize}
\item $\mathbf E[\mathbf{Var}[\mathcal N_3(t)| \mathcal N_1,
  \mathcal N_2] ]= \lambda_1^+\lambda_2\lambda_3\tau_1 \tau_2\tau_3$
\item $\mathbf E[\mathbf{Var}[\mathbf E[\mathcal N_3(t)| \mathcal
  N_1, \mathcal N_2]| \mathcal N_1] ]= \lambda_1^+\lambda_2
  \lambda_3^2 \tau_1 \tau_2 \tau_3^2 \cdot A$
\item $\mathbf{Var}[\mathbf E[\mathcal N_3(t)| \mathcal N_1]]
  =\lambda_1^+ \lambda_1^- \lambda_2^2 \lambda_3^2 \tau_1^2 \tau_2^2
  \tau_3^2 \cdot B$,
\end{itemize}
where $A$ and $B$ are described in \eqref{eq:T22} and \eqref{eq:T23},
respectively. These parts mirror the contribution of individual
protein-noise, individual RNA-noise and noise caused by gene
activation to the variance in protein number, respectively.  Such a
variance decomposition is well-known for the model without delay
\citep{Paulsson2005} and is helpful in understanding the different
kind of effects. Our results on this variance decomposition, togehter
with the concrete formulas for $A$ and $B$ seem to be the first
analytical solution of the delay model from \eqref{eq:model} for gene
expression.

~

It has been known for a long time that production of proteins (within
a system of active and deactive genes) comes in bursts. Although our
results on the first two moments give only a first impression about
this burst-like behavior, the connection is only indirect. We rely on
the intuition that a higher variance is compatible with a more
burst-like behavior of protein.  With this interpretation, we find
that burst-like behavior is highest without delay. This result can
also be explained intuitively: Delay weakens the discrete transitions
between the state of gene or the number of RNA
respectively. Consequently the RNA and protein expression patterns
tend to be less bursty given delay.

~

Analytical approaches for biochemical systems frequently utilize the
Master equation, available for any Markov process. Since delays lead
to non-Markovian processes, this technique has to be adapted in order
to capture delays. One way out -- e.g.\ carried out in
\cite{Bratsun:2005:Proc-Natl-Acad-Sci-U-S-A:16199522} and
\cite{Tian2007} -- is to use independence of two-point probabilities
in the Master equation in order to have a closed system of delay
differential equations. However, note that the resulting description
is not precise, whereas the model equations \eqref{eq:model1} provide
an exact description of delay stochastic systems; compare also with
the approach from \cite{Anderson:2007:J-Chem-Phys:18067349}.

~

Delays have been considered for gene expression in equations for
protein degradation and feedback in
\cite{Bratsun:2005:Proc-Natl-Acad-Sci-U-S-A:16199522} by lumping
transcription and translation into a single process. Protein
degradation is a process involving complex proteolytic pathways and a
cellular degradation machinery, leading to several delays; see also
\cite{Robert2013}. Moreover, transcription via elongation is known to
produce delays which are able to explain oscillatory behavior of
feedback systems
\citep{Monk:2003:Curr-Biol:12932324,Roussel:2006:Phys-Biol:17200603}
Interestingly, \cite{Bratsun:2005:Proc-Natl-Acad-Sci-U-S-A:16199522}
find examples with feedback where the stochastic system is oscillatory
even if the corresponding deterministic system is
non-oscillatory. Such oscillatory behavior was also modeled for the
expression levels of both RNA and protein of the Notch effector Hes1
by \cite{Barrio:2006:PLoS-Comput-Biol:16965175}.

~

A large part of theoretical work on delay models is dealing with
simulation schemes for delay stochastic equations; see the review by
\cite{Ribeiro:2010:Math-Biosci:19883665}.
\cite{Bratsun:2005:Proc-Natl-Acad-Sci-U-S-A:16199522} extend the
classical SSA method of \cite{Gillespie1977} (which ignores delay) by
keeping a list of reactions, which were initiated but finish only
later. Another approach is to allow for memory reactions and memory
species as used in \cite{Tian:2013:PLoS-One:23349679}.  According to
\cite{Barrio:2006:PLoS-Comput-Biol:16965175}, delay reactions must be
decomposed into consuming and non-consuming reactions. The reactants
in an unfinished nonconsuming reaction can already participate in a
new reaction, while they cannot participate in a consuming
reaction. (An example of the former is a new initiation of
transcription by binding of RNA polymerase can happen although the
last transcription is not finished yet.) In this sense, our model
\eqref{eq:model} uses consuming reactions. These simulation schemes
have been improved by schemes using a smaller number of random
variables by \cite{Cai:2007:J-Chem-Phys:17411109} and
\cite{Anderson:2007:J-Chem-Phys:18067349}. The latter approach uses
the Poisson process representation of delay models, much in the spirit
of our paper; compare with \eqref{eq:model1}.

~

Today, it is known that variances of protein numbers for the
expression of a multitude of genes is mainly based on RNA flucations
\citep{BarEven2006}. Clearly, this variance in protein number is also
affected by delays. Hence, we have to know the variances in delay
times for practical purposes, if we want to compare theoretical and
empirical fluctuations in protein numbers. On the empirical side,
measurements of delay time variances will be most important for
understanding delay on gene regulation. On the theoretical side, an
important extension of our theory would be to incorporate
self-regulatory mechanisms. Note that several authors found that
delays in such systems can lead to oscillatory behavior or even more
bursty behavior \citep{Monk:2003:Curr-Biol:12932324,zavala2014delays,
  Bratsun:2005:Proc-Natl-Acad-Sci-U-S-A:16199522,zavala2014delays}.
Describing such feedbacks using point processes will be a more
thorough understanding of delays in gene regulation.



\begin{appendix}
  \bigskip
  \bigskip
  \noindent {\huge \bf Appendix}
  \section{A key lemma}
  Within this section, we summarize some frequently used computations
  in the following lemma.

  \begin{lemma}
    Let \label{l:key} $T,T'$ be independent exponentially distributed
    with expectation $\sigma$, and $G, G'\geq 0$ be independent and
    identically distributed. Then, for $t\in \mathbb R, \tau>0$
    \begin{align}
      \label{eq:key1}
      & \mathbf E\Big[e^{-\frac{|t+T+G-T'-G'|}{\tau}}\Big] =
      \frac{\tau}{\tau+\sigma} \frac{\tau \mathbf
        E\Big[e^{-\frac{|G'-G-t|}{\tau}}\Big] - \sigma\mathbf
        E\Big[e^{-\frac{|G'-G-t|}{\sigma}}\Big]}{\tau-\sigma}.
      \intertext{Moreover, if $S$ has a symmetric distribution,
        i.e. $S\stackrel d = -S$,}
      \label{eq:key2} 
      & \mathbf E\Big[e^{-\frac{||T+G-T'-G'|-S|}{\tau}}\Big] = \mathbf
      E\Big[e^{-\frac{|T+G-T'-G'+S|}{\tau}}\Big] \\ & \qquad \qquad
      \qquad \qquad \qquad \notag = \frac{\tau}{\tau+\sigma}
      \frac{\tau \mathbf E\Big[e^{-\frac{|G'-G + S|}{\tau}}\Big] -
        \sigma \mathbf E\Big[e^{-\frac{|G'-G +
            S|}{\sigma}}\Big]}{\tau-\sigma}.
    \end{align}
  \end{lemma}
  
  \begin{remark}
    Note that in \eqref{eq:key1} and \eqref{eq:key2} all left hand
    sides include the exponentially distributed random variables $T,
    T'$, while the right hand sides only depend on the distribution of
    $G, G'$ (and $S$).
  \end{remark}
  
  \begin{proof}
    We start with the proof of \eqref{eq:key1}.  First,
    \begin{align} \label{eq:key4} \mathbf E\Big[ &
      e^{-\frac{|t+T+G-T'-G'|}{\tau}}\Big] = \mathbf E\Big[\mathbf
      E\big[e^{-\frac{t+T+G-T'-G'}{\tau}}1_{T>T'+G'-G-t}\big|G-G'\big]\Big]
      \\ & \qquad \qquad \qquad \qquad \qquad + \mathbf E\Big[\mathbf
      E\big[e^{-\frac{T'+G'-T-G-t}{\tau}}1_{T'\geq
        T+G-G'+t}\big|G-G'\big]\Big] \notag
    \end{align}
    Then,
    \begin{align}\label{eq:keycomb}
      \mathbf E & \big[e^{-\frac{t+T+G-T'-G'}{\tau}}
      1_{T>T'+G'-G-t}\big|G-G'\big] \\ & \notag = 1_{G'-G-t \geq 0}
      \mathbf
      E\big[e^{-\frac{t+T+G-T'-G'}{\tau}}1_{T>T'+G'-G-t}\big|G-G'\big]
      \\ \notag & \qquad \qquad \qquad + 1_{G'-G-t < 0} \mathbf
      E\big[e^{-\frac{t+T+G-T'-G'}{\tau}} \big(1_{T'<G-G'+ t} \\ &\notag
      \qquad \qquad \qquad \qquad \qquad \qquad \qquad + 1_{T'\geq
        G-G'+t}1_{T>T'+G'-G-t}\big) \big|G-G'\big].
    \end{align}
    Now, on the set $\{G'- G-t\geq 0\}$,
    \begin{align*}
      \mathbf E\big[&
      e^{-\frac{t+T+G-T'-G'}{\tau}}1_{T>T'+G'-G-t}\big|G-G'\big] = \frac
      12 e^{-\frac{G'-G-t}{\sigma}} \mathbf
      E\big[e^{-\frac{T}{\tau}}\big] \\ & = \frac 12
      \frac{\tau}{\sigma+\tau}e^{-\frac{G'-G-t}{\sigma}}
    \end{align*}
    since $\mathbf P[T>T'+G'-G-t|G - G'] = \tfrac 12
    e^{-\frac{G'-G-t}{\sigma}}$ and given $T>T'+G'-G - t$ and
    conditioned on $G-G'$, the random variable $t + T+G-T'-G'$ is
    again exponentially distributed with expectation $\sigma$. Next,
    on $\{G'-G-t<0\}$
    \begin{align*}
      \mathbf E \big[ e^{-\frac{t+T+G-T'-G'}{\tau}} & 1_{T'<G-G'+t}
      \big|G-G'\big] \\ & = \int_0^\infty \int_0^{G-G'+t}
      \frac{1}{\sigma^2}e^{-\frac{s'+s}{\sigma}}
      e^{-\frac{t+s+G-s'-G'}{\tau}} ds' ds \\ & =
      \frac{1}{\sigma^2}e^{-\frac{G-G'+t}{\tau}} \int_0^\infty
      e^{-s\big(\frac{1}{\sigma}+\frac{1}{\tau}\big)}ds
      \int_0^{G-G'+t}
      e^{-s'\big(\frac{1}{\sigma}-\frac{1}{\tau}\big)}ds' \\ & =
      \frac{1/\sigma}{1/\sigma + 1/\tau}\frac{1/\sigma}{1/\sigma -
        1/\tau} e^{-\frac{G-G'+t}{\tau}} \big(1 -
      e^{-(G-G'+t)\big(\frac{1}{\sigma}-\frac{1}{\tau}\big)}\big) \\ &
      = \frac{\tau^2}{(\tau+\sigma)(\tau-\sigma)}
      \big(e^{-\frac{G-G'+t}{\tau}} - e^{-\frac{G-G'+t}{\sigma}} \big)
    \end{align*}
    as well as
    \begin{align*}
      \mathbf E & \big[ e^{-\frac{t+T+G-T'-G'}{\tau}} 1_{T'\geq
        G-G'+t}1_{T>T'+G'-G-t}\big|G-G'\big] \\ & =
      e^{-\frac{t+G-G'}{\sigma}} \mathbf
      E\big[e^{-\frac{T-T'}{\tau}}1_{T>T'}|G-G'\big] = \frac 12
      e^{-\frac{t+G-G'}{\sigma}}\frac{\tau}{\sigma+\tau}
    \end{align*}
    since $\mathbf P[T'>G-G'+t|G-G'] = e^{-\frac{t+G-G'}{\sigma}}$ and
    given $T'>G-G'+t$, the random variable $T'+G'-G-t$ is again
    exp$(\sigma)$ distributed. Plugging the last three computations
    into~\eqref{eq:keycomb},
    \begin{align*}
      \mathbf E & \big[e^{-\frac{t+T+G-T'-G'}{\tau}}
      1_{T>T'+G'-G-t}\big|G-G'\big] \\ & = \frac 12 \frac{\tau}{\sigma
        + \tau} e^{-\frac{|G'-G-t|}{\tau}} + 1_{G'-G-t<0}
      \frac{\tau^2}{(\tau+\sigma)(\tau-\sigma)}
      \notag\big(e^{-\frac{G-G'+t}{\tau}} - e^{-\frac{G-G'+t}{\sigma}}
      \big).
    \end{align*}
    Then, by symmetry of $G-G'$, from~\eqref{eq:key4},
    \begin{align*}
      \mathbf E & \big[e^{-\frac{|t+T+G-T'-G'|}{\tau}}\big] \\ & =
      \mathbf E\Big[\frac 12 \frac{\tau}{\sigma + \tau}
      e^{-\frac{|G'-G-t|}{\sigma}} + 1_{G'-G-t<0}
      \frac{\tau^2}{(\tau+\sigma)(\tau-\sigma)}
      \big(e^{-\frac{G-G'+t}{\tau}} - e^{-\frac{G-G'+t}{\sigma}}
      \big)\Big] \\ & \quad + \mathbf E\Big[\frac 12
      \frac{\tau}{\sigma + \tau} e^{-\frac{|G-G'+t|}{\sigma}} +
      1_{G-G'+t<0} \frac{\tau^2}{(\tau+\sigma)(\tau-\sigma)}
      \big(e^{-\frac{G'-G-t}{\tau}} - e^{-\frac{G'-G-t}{\sigma}}
      \big)\Big] \\ & = \frac{\tau}{\sigma + \tau} \mathbf
      E\Big[e^{-\frac{|G'-G-t|}{\sigma}}\Big] +
      \frac{\tau^2}{(\tau+\sigma)(\tau-\sigma)} \mathbf
      E\Big[e^{-\frac{|G'-G-t|}{\tau}} -
      e^{-\frac{|G'-G-t|}{\sigma}}\Big] \\ & =
      \frac{\tau}{\tau+\sigma} \frac{\tau \mathbf
        E\Big[e^{-\frac{|G'-G-t|}{\tau}}\Big] -\sigma \mathbf
        E\Big[e^{-\frac{|G'-G-t|}{\sigma}}\Big]}{\tau-\sigma}.
    \end{align*}
    For, \eqref{eq:key2}, we only need to show the first equality
    since the second follows from \eqref{eq:key1} by conditioning on
    $S$. Using the symmetry of $S$ we can write
    \begin{align*}
      \mathbf E\Big[e^{-\frac{||T+G-T'-G'|-S|}{\tau}}\Big] & = \mathbf
      E\Big[e^{-\frac{|T+G-T'-G'-S|}{\tau}}, T+G-T'-G'\geq 0\Big] \\ &
      \qquad \qquad \qquad + \mathbf
      E\Big[e^{-\frac{|T'+G'-T-G+S|}{\tau}}, T+G-T'-G'< 0\Big] \\ & =
      \mathbf E\Big[e^{-\frac{|T+G-T'-G'-S|}{\tau}}\Big]
    \end{align*}
    which shows the first equality in \eqref{eq:key2}.
  \end{proof}
\end{appendix}

\subsection*{Acknowledgments}
We thank Bence Melykuti for helpful comments on the manuscript.

\bibliographystyle{chicago}
\bibliography{expression}

\end{document}